\documentclass[a4paper,twocolumn,11pt,accepted=2026-09-04]{quantumarticle}
\pdfoutput=1
\usepackage[english]{babel}
\usepackage{graphicx}
\usepackage[caption=false]{subfig}
\usepackage{mathrsfs}
\usepackage{amsmath}
\usepackage{amssymb}
\usepackage{amsthm}
\usepackage{bbold}
\usepackage{thmtools, thm-restate}
\usepackage{bm}
\usepackage{xcolor}
\usepackage{hyperref}
\usepackage{hypernat}
\usepackage[utf8]{inputenc}
\usepackage[T1]{fontenc}
\usepackage{cleveref}
\usepackage{soul}
\usepackage[backend=bibtex]{biblatex}
\usepackage{braket}
\newcommand{\id}{\mathbb{1}}

\DeclareMathOperator{\Tr}{Tr}

\usepackage{tikz}
\usepackage{lipsum}

\declaretheorem[parent=section,name=Theorem]{thm}
\declaretheorem[sibling=thm,name=Lemma]{lem}
\declaretheorem[style=definition,sibling=thm]{definition}

\declaretheorem[sibling=thm]{proposition}

\begin{document}

\title{Generalising Aumann's Agreement Theorem}

\author{Matthew Leifer}
\affiliation{Schmid College of Science and Technology, Chapman University, One University Dr., Orange, CA 92866, USA}
\affiliation{Institute for Quantum Studies, Chapman University, One University Dr., Orange, CA 92866, USA}
\author{Cristhiano Duarte}
\affiliation{Schmid College of Science and Technology, Chapman University, One University Dr., Orange, CA 92866, USA}
\affiliation{Institute for Quantum Studies, Chapman University, One University Dr., Orange, CA 92866, USA}
\affiliation{Instituto de Física, Universidade Federal da Bahia, Campus de Ondina, Rua Barão do Geremoabo, s.n., Ondina, Salvador, BA 40210-340, Brazil}%
\affiliation{Universidade Federal de Juiz de Fora, Departamento de Física, Juiz de Fora, MG, Brasil}

\begin{abstract}
According to Aumann's celebrated theorem, rational agents cannot agree to disagree. In other words, agents who once shared a common prior probability distribution and who have common knowledge about their posteriors cannot assign different probability distributions to a given proposition. Common knowledge imposes strong restrictions on assigned probabilities. In fact, Aumann's agreement theorem was one of the first attempts to formalise and explore the role played by common knowledge in decision theory. Recently, the debate over possible (quantum) extensions of Aumann's results has resurfaced. This paper contributes to this discussion. First, we argue that agreeing to disagree is impossible in quantum theory. Secondly, by building on the quantum argument, we show that agreeing to disagree is also forbidden in any generalised probability theory. The upshot is that in its probabilistic version, the agreement theorem is a direct consequence of how we choose to condition upon acquiring new information.   
\end{abstract}

\maketitle


\section{Introduction}\label{Sec.Introduction}

Is it possible to agree to disagree? If we consider the original version of J. R. Aumann's intriguing and seminal result \cite{Aumann76}, the answer is a resounding: \textit{'no, it is not!'}. According to Aumann's theorem, whenever a family of agents reach common knowledge about the description of an event, there is no escape; they all have to agree with each other about the description of that event---provided they started from the same prior. We will unpack the assumptions and hypotheses underlying the theorem in the subsequent sections. For now, it suffices to keep the simpler and condensed version in mind: rational agents cannot agree to disagree.

Aumann's impossibility theorem is provocative for various reasons. Firstly, it explores the close connection between two intricate, overlapping, and paradigmatic concepts in epistemology: knowledge \cite{Pritchard13} and relativism \cite{Pritchard19, Rovelli18, BJ21}. Essentially, the theorem states that a strong notion of shared knowledge within a group implies that agents in that group must paint their world with the same colours. Some strands of relativism are immediately ruled out whenever common knowledge holds true---because there cannot be rational agreements on disagreements. Secondly, the impossibility theorem opens up the possibility of distinguishing rational agents from non-rational agents. For example, in more dynamical derivations of Aumann's theorem, where public announcements play a fundamental role \cite{GP82,Demey14,Demey14b,DHK08}, after finitely many rounds of truthful exchange of information, if the agents are still in disagreement, it may be the case that at least one of them is not acting rationally\footnote{This may explain why online debates can go on for a long time with no hint of agreement to be found.}. In this sense, rather than using a betting system \cite{Vineberg16}, one could, in theory, use a collective property to define rationality normatively. Finally, Aumann's theorem may only apply to the world because classical probability provides a coarse-grained representation of nature.  More general probabilistic theories may introduce additional possibilities for agents to disagree with one another.  In this vein, it is worth recalling that quantum theory can be viewed as a generalisation of classical probability theory \cite{Wilce24} and that "entanglement is a trick quantum magicians use to produce phenomena that classical magicians cannot imitate" \cite{Bruss02}.

In this work, we explore extensions to the impossibility of agreeing to disagree in hybrid, post-classical scenarios. We demonstrate that Aumann's theorem is not exclusive to classical probability theory but extends to any generalised probabilistic description of nature, provided that notions of knowledge and conditioning are well-defined. Our hybrid scenarios mix post-classical descriptions with a classical definition for (propositional) knowledge---what is required for an agent to know that something. As will become clear later, in such scenarios, it is the set-theoretical notion of common knowledge as well as our notion of conditioning that do all the heavy lifting in proving the impossibility theorem.  

There have been other attempts to capture the content of Aumann's agreement theorem in quantum and even in post-quantum frameworks before. In some of them, Aumann's theorem breaks apart and could, in principle, be regarded as a sign that quantum theory, seen as a theory of probabilistic assignments, is more resourceful than its classical counterpart \cite{CTGAAP21,BrandenburgerEtAl24, Khrennikov15, KB14}. Contrary to those attempts, our results will point in the other direction. We will argue that for specific notions of conditioning and knowledge models, Aumann's impossibility theorem remains valid regardless of the toolbox rational agents use to describe their uncertainties.

Our manuscript is structured as follows. To facilitate the reading, in Sec.\ref{Sec.OriginalArgument}, we review Aumann's original argument for the impossibility of agreeing to disagree. Although there are more modern approaches to Aumann's impossibility result, we decided to stick to the original argument, as it requires less effort on a first reading. Sec.\ref{Sec.QuantumArgument} contains our first main result. There we prove that a hybrid quantum version of the agreement theorem also holds in a quantum-like scenario. We also introduce all the necessary elements for the proof in that section. Inspired by the quantum reformulation, we give a similar argument for Generalized Probability Theories (GPTs) in sec.\ref{Sec.GPTArgument}, showing that Aumann's theorem can also be extended to GPTs. Sec.\ref{Sec.Limitations} clarifies the main limitations and hypotheses we have used throughout this contribution.  In Sec.\ref{Sec.Comparison}, we compare our findings with other similar results in the literature and hint at possible future works. We conclude our work in Sec.\ref{Sec.Conclusion}. 


\section{Aumann's Original Argument}\label{Sec.OriginalArgument}

This section reviews Aumann's original argument for the impossibility of agreeing to disagree, a fundamental result in decision theory. Our primary objective is to introduce the theorem and highlight the mathematical formalism involved in the agreement result, as it will be used repeatedly in the subsequent sections. We do not intend to provide a comprehensive account of the agreement theorem, nor do we want to present an in-depth discussion of all the subsequent works \cite{RW90, BN97}, modifications \cite{Samet10, Samet20, DR12} and debates \cite{Aaronson05, Lederman15} that have followed up on Aumann's seminal work. We refer to \cite{Demey14} for an in-depth overview of the topic. 

Aumann was responsible for one of the first attempts to rigorously explore the notion of common knowledge in decision theory. In plain words, Aumann's agreement theorem \cite{Aumann76} says that whenever a group of agents had agreed in the past, if their current description of a proposition is common knowledge among them, then their descriptions must match each other's. Slightly more precisely, if two agents $A$ and $B$ have the same priors, it is impossible that it is common knowledge among the agents that $A$ assigns to some event a probability $a$, and $B$ assigns to the same event a probability $b$ with $a \neq b$. Statically, one can write down the main logical character of Aumann's impossibility theorem as follows: 
\begin{equation}
\mbox{[equalpriors}\land\mbox{C(posteriors)]} \rightarrow \mbox{equalposteriors}, 
\label{Eq.InformalLogicAgreement}
\end{equation}
where $C$ plays the role of a common knowledge operator. We will see in a minute that demanding common knowledge imposes strong restrictions on the set of events and consequently restricts the possible probability assignments available to the agents. 

Before we delve into the agreement theorem and its implications, as well as the generalisations we aim to make, it is crucial to establish some definitions first. To discuss knowledge and common knowledge, which are influential by-products of Aumann's work \cite{Aumann76}, we need a solid translation of these concepts into a mathematical model. It will become clear that Aumann's theorem follows directly from this mathematical model for knowledge that we now describe. 

\begin{definition}[Knowledge Model]
A \emph{knowledge model} for $N \in \mathbb{N}$ agents is a structure $(\Omega, Q_1, Q_2,...,Q_N, \Sigma)$ consisting of
\begin{enumerate}
    \item [(i)] a non-empty and finite set $\Omega$ of states of the world,
    \item[(ii)] a partition $Q_i$ of $\Omega$ for each $i \in [N]$, and
    \item[(iii)]a $\sigma-$algebra $\Sigma$ over $\Omega$ that includes all the elements of $Q_1,Q_2,...,Q_N$.
\end{enumerate}
   We usually denote $Q_i(\omega)$, or $Q_iw$, as that unique element of the partition $Q_i$ containing $\omega \in \Omega$.
\label{Def.KnowledgeModel}   
\end{definition}

A knowledge model $(\Omega, Q_1, Q_2,...,Q_N, \Sigma)$ mimics a situation where $N$ individuals are about to learn the answer to various questions---possibly by observing the outcomes of tests, experiments, horse races and the like. The answers for agent $i$'s questions are codified in the partition $Q_i=\{Q_{i}^1,Q_{i}^2,...,Q_{i}^{M(i)}\}$ representing mutually exclusive, collectively exhaustive propositions. See fig. \ref{Fig.KnowledgeModel}. 

\begin{figure}
    \centering
    \includegraphics[scale=0.25]{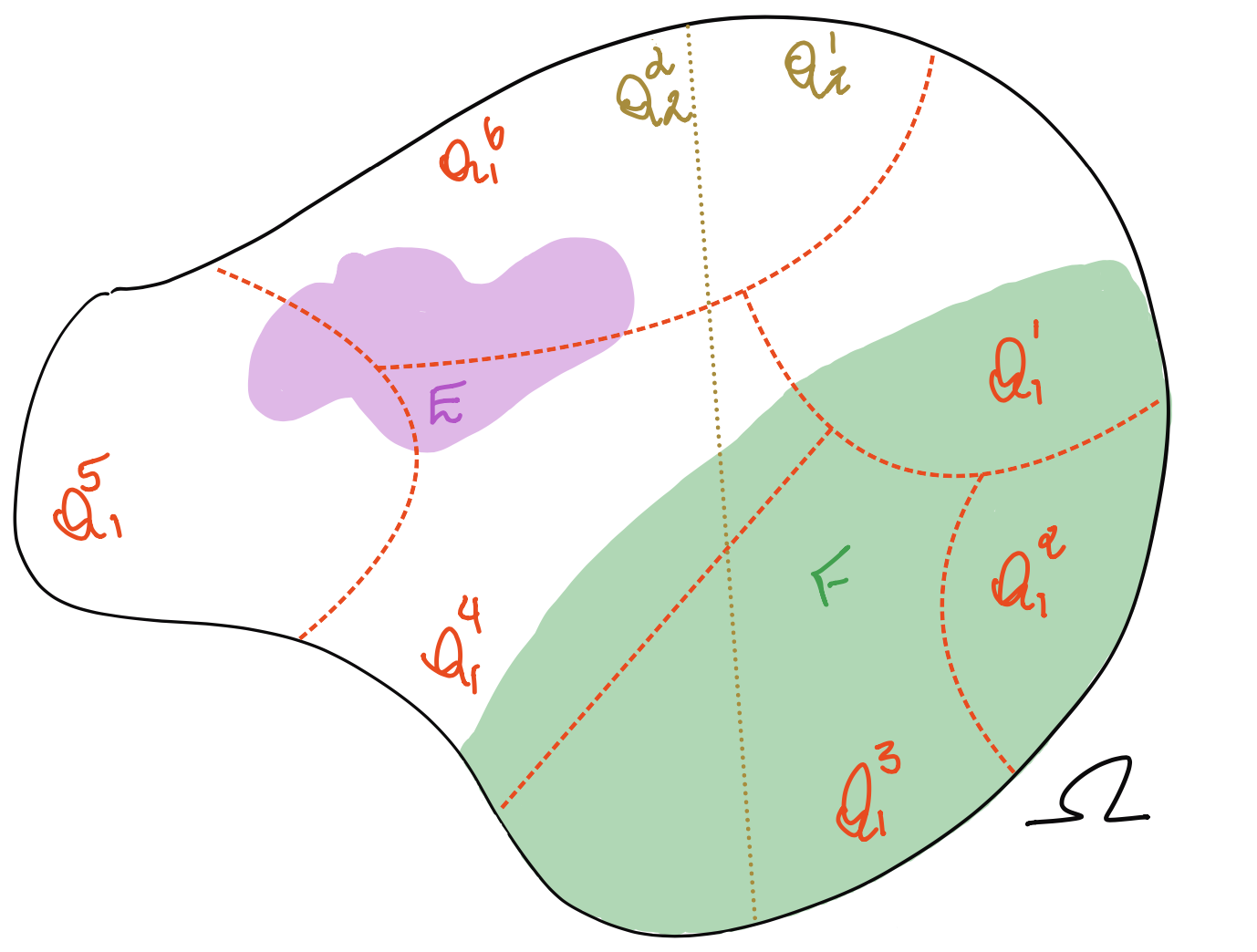}
    \caption{Knowledge model for two agents: yellow and orange (colours online). Regardless of the state of the world $\omega \in \Omega$, neither agent can know the purple $E$ event. Similarly, only the orange agent can know the green $F$ event.}
    \label{Fig.KnowledgeModel}
\end{figure}

\begin{definition}[Knowledge] Let $(\Omega, Q_1,$ $Q_2,...,Q_N, \Sigma)$ be a knowledge model. We say that an agent $i$ \emph{knows} $E \in \Sigma$ at $\omega$ whenever $Q_{i}(\omega) \subseteq E$. The $i$-th knowledge operator for this model is defined as
\begin{align}
    K_{i}:\Sigma &\rightarrow \Sigma \nonumber \\
     E &\mapsto K_{i}(E)=\{\omega \in \Omega| Q_{i}(\omega) \subseteq E\}. 
    \label{Eq.DefKnowledgeOperator}
\end{align}
\label{Def.Knowledge}
\end{definition}  

Considering the reasoning of only one agent, in a knowledge model, their knowledge about an event $E$ is the set of all possible states of the world where that agent knows $E$. Figure 1 schematically represents the knowledge model for two agents, Yellow and Orange. Because of their partition, neither agent can know the purple event $E$. On the other hand, only the Orange agent can know the event $F$—Yellow's partition is too coarse-grained for this agent to know (that) $F$.

Aumann's agreement theorem is about a more refined notion of knowledge that involves multiple agents: common knowledge. We will address common knowledge in a minute, but first, we need to define an infinite hierarchy of group or mutual knowledge \footnote{In Kripke models, there is a non-equivalent way to introduce the notion of common knowledge among agents. It is beyond the scope of this manuscript to discuss the differences between the two definitions, as well as the subtleties associated with the latter involving Kripke models. We refer to \cite{Demey14} for an introduction.}. Common knowledge will be the limiting case of this nested definition.   

\begin{definition}[Mutual and Common Knowledge]
Let $(\Omega, Q_1, Q_2,...,Q_N, \Sigma)$ be a knowledge model and $E$ an arbitrary event in $\Sigma$. The \emph{mutual knowledge of $m-$th degree} of $E$ is recursively defined as:
\begin{align}
      &\mbox{0-th}:M_{0}(E):=E \nonumber \\
      &\mbox{1-st}:M_1(E):=K_1(E)\cap K_{2}(E) \cap ... \cap K_{N}(E) \nonumber \\
      &\mbox{2-nd}:M_2(E):=K_1(M_1(E)) \cap ... \cap  K_{N}(M_1(E)) \nonumber \\  
     \vdots  \nonumber\\
     & m\mbox{-th}:M_{m}(E):= \bigcap_{i=1}^{N}K_{i}(M_{m-1}(E)). 
     \label{Eq.DefMutualKnowledge}
\end{align}
Similarly, \emph{common knowledge} is defined as 
\begin{equation}
    C(E):= \bigcap_{m \in \mathbb{N}}M_{m}(E).
    \label{Eq.DefCommonKnowledge}
\end{equation}
\label{Def.CommonKnowledge}
\end{definition}

For simplicity, consider the situation involving only two agents. In this case, mutual knowledge of second-order accounts for two situations. First, $A$ knows that she knows $E$, and she also knows that $B$ knows $E$. Second, and symmetrically, $B$ knows that he knows $E$ and also that he knows that $A$ knows $E$.  What is crucial here is how mutual knowledge (of a finite degree) differs from common knowledge. Common knowledge between two agents means that an infinite list of $A$ knows that $B$ knows that $A$ knows that $B$ knows that... and so on holds true.    

As thoroughly discussed in \cite{Lederman15}, common knowledge is the key concept in Aumann's theorem. In the knowledge model we are dealing with here, the possibility of an event that the agents can commonly know puts strong constraints on the possible partitions allowed by the model. In contrast to the more dynamic formulation explored in \cite{Demey14b}, where common knowledge is always reached through dialogue among the agents, Aumann's formulation is static and fixed. As a result, depending on the sets $Q_1, Q_2, ..., Q_N$ the agents start with, it may well be the case that there is no $C(E)$ representing their common knowledge. The lemma below makes this affirmation more precise.  

\begin{lem}
Let $(\Omega, Q_1, Q_2,...,Q_N, \Sigma)$ be a knowledge model and $E$ an event. If $C(E) \neq \emptyset$, then for each $i \in [N]$ there exists a finite family $\{D_{i}^{1},...,D_{i}^{k_i}\} \subseteq  Q_i$ such that 
\begin{enumerate}
    \item [(i)] $C(E)=\bigcup_{l \in [k_i]}D_{i}^{l}$, and
    \item [(ii)] $D_{i}^{l} \cap D_{i}^{l'} = \emptyset$,
\end{enumerate}
for every $i \in [N]$ and for every $l \neq l' \in [k_{i}]$
\label{Lemma.SetTheoryLemma}
\end{lem}
\begin{proof}
We only study the inclusion $C(E) \subseteq \bigcup_{l \in [k_i]}D_{i}^{l}$. The other direction follows directly from definition \ref{Def.CommonKnowledge}. Suppose that $\omega \in C(E)$. In this case:
\begin{align}
    &\Longrightarrow \omega \in M_{m}(E),\quad \forall m \in \mathbb{N} \nonumber \\
    &\Longrightarrow \omega \in E, \, Q_{i}(\omega) \subseteq E,  \, Q_{i}(\omega) \subseteq \bigcap_{i=1}^{N}K_{i}(E), \, ...  \nonumber \\
    &\Longrightarrow Q_{i}(\omega) \subseteq M_{m}(E), \quad \forall m \Longrightarrow Q_{i}(\omega) \subseteq C(E).
\end{align}
As we are considering only a finite set of states of the world (see def. \ref{Def.KnowledgeModel}), $C(E)$ is also finite. In this case, with no loss of generality, we can assume that $C(E)=\{e_1,\dots,e_{|C(E)|}\}$. Now, note that $\{Q_{i}(e_1),\dots, Q_{i}(e_{|C(E)|})\}$ has at most $|C(E)|$ disjoint elements, and that it is exactly its non-repeating elements we will use to form the set $\{D_{i}^{1},\dots, D_{i}^{k_i}\}$.
\end{proof}

As we anticipated in the preceding paragraph, before lemma \ref{Lemma.SetTheoryLemma}, the existence of common knowledge in the knowledge model imposes strong restrictions on the structure of the outcome set $Q_1,..., Q_N$, and therefore, not every event can be common knowledge among the agents. Also, recalling def. \ref{Def.Knowledge}, note that the equality demanded by lemma \ref{Lemma.SetTheoryLemma} imposes that the common knowledge set is also known by all of the agents in the model.

Now that we have defined what common knowledge is and how it can be characterised in terms of the partitions $Q_1,...,Q_N$ in a given knowledge model, we have all the ingredients to address Aumann's agreement theorem.

\begin{thm}[Aumann's Agreement Theorem]
Let $(\Omega, Q_1, Q_2,...,Q_N, \Sigma)$ be a knowledge model and $\mathbb{P}:\Sigma \rightarrow [0,1]$ be a probability function over 
$\Omega$. Define 
\begin{equation}
E:=\cap_{i \in [N]} \{\omega \in \Omega | \,\, \mathbb{P}(H|Q_{i}(w))=q_i \},
\label{Eq.ThmClassicalCommonKnowledgeEvent}
\end{equation}
with $H \in \Sigma$ and $q_1,q_2,...,q_N \in [0,1].$ If $\mathbb{P}(C(E)) \neq 0$, then \begin{equation}q_1=q_2=...=q_N=\mathbb{P}(H|C(E)).\end{equation}
\label{Thm.ClassicalAumannTheorem}
\end{thm}
\begin{proof}
\begin{align}
 &\mathbb{P}(H|C(E))=\frac{\mathbb{P}(H \cap C(E))}{\mathbb{P}(C(E))}= \frac{\mathbb{P}(H \cap \bigcup_{j \in [k_i]}D_{i}^{j})}{\mathbb{P}(\bigcup_{j \in [k_i]}D_{i}^{j})} \\
 &= \frac{\sum_{j \in [k_i]}\mathbb{P}(H | D_{i}^{j}) \mathbb{P}(D_{i}^{j})}{\sum_{j \in [k_i]}\mathbb{P}(D_{i}^{j})} = \frac{\sum_{j \in [k_i]} q_i  \mathbb{P}(D_{i}^{j})}{\sum_{j \in [k_i]}\mathbb{P}(D_{i}^{j})}  \\
 &= q_i. 
 \label{Eq.ThmClassicalAumann}
\end{align}
As the argument is valid for each agent $i$, it follows that 
\begin{equation}
    q_1=q_2=...=q_N=\mathbb{P}(H|C(E)).
\end{equation}
\end{proof}

In words, Aumann's theorem says that for a given event $H$, which all the agents want to assign a probability to, if the individual probabilities each agent assigns to $H$ are common knowledge, then these individual probabilities must be the same—even if they have been obtained by completely different observations and experiences. Because $C(E)$ is defined in the underlying knowledge model, without mentioning the probabilistic layer superimposed on it, a few words about $P(C(E))$ may clarify the role it plays in the theorem. Assuming that $P(C(E))$ is non-zero is equivalent to saying that $C(E) \neq \emptyset$, or that there exists at least a state of the world $\omega \in \Omega$ such that $(a)$ in this state of the world, every agent $i$ assigns $q_i$ to $H$; and $(b)$ every agent $i$ knows that every other agent $j$ assigns $q_j$ to $H$; and $(c)$ every agent knows that every agent knows what every agent assigns to $H$ and so on. It is as if this hierarchy of knowledge is fully spread among the agents. Consequently, if one could see each level of the hierarchy truthfully and publicly announced to every other agent, in the end those agents would be led to adapt their probabilistic assignments to match each other's—although we do not explain how these dynamical steps take place, see ref. \cite{Demey14b} for a dynamical turn on Aumann's theorem. 

We want to emphasise that it is both the structure of the common knowledge set $C(E)$ and the notion of conditioning we adopt here that force a common assignment across all agents for the event `$H$ given $C(E)$'. The same structure is replicated in our generalisations of the agreement theorem.

There is another point worth noting. The event $E$ defined in eq. \eqref{Eq.ThmClassicalCommonKnowledgeEvent} depends on the particular choice of $H \in \Sigma$, and so writing the posteriors $q_1,q_2,...,q_N$ with no explicit mention of $H$ is a slight abuse of notation. It would have been more precise if we had written $E(H)$ for the event in \eqref{Eq.ThmClassicalCommonKnowledgeEvent}, and $q_1(H),...,q_N(H)$ for the list of agents' posterior probability assignments. Obviously, the agreement theorem also holds true if we make this functional dependence more explicit---it follows as a direct consequence of Aumann's original setup.   

\begin{thm}[Aumann's Theorem - Second Version]
Let $(\Omega, Q_1, Q_2,...,Q_N, \Sigma)$ be knowledge model and $\mathbb{P}:\Sigma \rightarrow [0,1]$ be a probability function over 
$\Omega$. For each $H \in \Sigma$, define: 
\begin{equation}
E_{H}:=\cap_{i \in [N]} \{\omega \in \Omega | \,\, \mathbb{P}(H|Q_{i}(w))=q_i(H) \},
\label{Eq.ThmClassicalCommonKnowledgeEventv2}
\end{equation}
with $q_1,q_2,...,q_N: \Sigma \rightarrow [0,1].$ If for each $H$ in $\Sigma$ $\mathbb{P}(C(E_H)) \neq \emptyset$, then \begin{equation}q_1=q_2=...=q_N=\mathbb{P}( \cdot |C(E_{H})).\end{equation}
\label{Thm.ClassicalAumannTheoremv2}
\end{thm}

Thm. \ref{Thm.ClassicalAumannTheoremv2} is a stronger formulation of Aumann's original theorem. Nonetheless, it still states that agents with the same priors cannot agree to disagree on the basis of their common knowledge of their posteriors. We have only made the functional dependence on $H$ more explicit, emphasising the validity of Aumann's theorem for every $H$ event \footnote{See the work in \cite{Demey14} for a discussion about the role played by the prior probability distribution and about the extent to which Aumann's theorem is always true, provided that the agents follow a precise notion of dialogue.}. In fact, we have only done so because this is exactly the format that creates the best parallel with our generalisations of the agreement theorem. 

\section{Quantum Version of the Agreement Theorem}\label{Sec.QuantumArgument}  

This section presents a hybrid quantum version of the agreement theorem. We prove that when the agents describe their uncertainty about a quantum system using a quantum state as opposed to a classical probability distribution, common knowledge of their quantum posteriors will always lead to an impossibility of agreeing to disagree. The present section also makes explicit the structure in Aumann's argument that is central not only for the original result but also for any potential generalisation of it.  

This is a hybrid quantum-classical generalisation because we are still using the classical knowledge model of def. \ref{Def.KnowledgeModel}. To a certain extent, we are granting the agents the possibility of expanding their reasoning abilities, as they are allowed to go beyond standard probability theory and use density operator valued measures (DOVMs) as their toolbox of analysis---even though their knowledge about these systems is in the form of classical data/classical set-theoretical models. We start by defining what a DOVM is.

\begin{definition}[DOVM]
Let $\mathcal{M}=(\Omega, \Sigma)$ be a measurable space and $\mathcal{H}$ a finite-dimensional Hilbert space. Let $\mathcal{L}(\mathcal{H}):=\{\rho: \mathcal{H} \rightarrow \mathcal{H}; \,\, \rho \mbox{ is linear}\}$ be the set of linear operators over $\mathcal{H}$. Finally, denote $\mathcal{D}(\mathcal{H}):=\{\rho \in \mathcal{L}(\mathcal{H}); \,\, \rho \geq 0 \mbox{ and } \Tr(\rho)=1\}$ the set of density operators over $\mathcal{H}$. A density operator valued measure (DOVM) over $\mathcal{M}$ is a map $\rho: \Sigma \rightarrow \mathcal{L}(\mathcal{H})$ such that:
\begin{itemize}
    \item [(i)] $\rho(\Omega)$ is a density operator in $\mathcal{D}(\mathcal{H})$; 
    \item [(ii)] $\rho(\Lambda) \geq 0$, for all $\Lambda \in \Sigma$;
    \item [(iii)] $\rho\left( \cup_{j \in J} \Lambda_{j} \right)= \sum_{j \in J}\rho(\Lambda_{j})$, for any countable family $\{\Lambda_{j}\}_{j \in J}$ of disjoint subsets in $\Sigma$. 
\end{itemize}
\label{Def.DOVMs}
\end{definition}

\begin{definition}[Conditional State]
Let $\rho: \Sigma \rightarrow \mathcal{L}(\mathcal{H})$ be a DOVM. For each $\Lambda \in \Sigma$, we define the object
\begin{equation}
    \rho_{|\Lambda}:=\frac{\rho(\Lambda)}{\mbox{Tr}[\rho(\Lambda)]}
    \label{Eq.DefConditionalFromDOVM}
\end{equation}
as the conditional state of the DOVM $\rho: \Sigma \rightarrow \mathcal{L}(\mathcal{H})$ with respect to the event $\Lambda$.
\label{Def.ConditionalFromDOVM}
\end{definition}

The following two propositions show the relationship between DOVMs and Positive Operator Valued Measures (POVMs), which are employed to describe measurements in quantum theory. 

\begin{proposition}
Given a DOVM $\rho: \Sigma \rightarrow \mathcal{L}(\mathcal{H})$, the following map is a POVM on supp$(\rho(\Omega))$:
\begin{equation}
    E :: \Lambda \mapsto \rho(\Omega)^{-1/2}\rho(\Lambda)\rho(\Omega)^{-1/2}
\end{equation}
\label{Prop.FromDOVMtoPOVM}
\end{proposition}
\begin{proof}
Without a loss of generality, we will assume that $\rho(\Omega)$ has full rank. The fact that E is positive and that $E(\Omega)=\id$ follows directly from the definition: 
\begin{align}
    \langle& \psi, E(\Lambda) \psi \rangle = \langle \psi, \rho(\Omega)^{-1/2}\rho(\Lambda)\rho(\Omega)^{-1/2} \psi \rangle \nonumber \\
    &= \langle \rho(\Omega)^{-1/2} \psi, \rho(\Lambda)\rho(\Omega)^{-1/2} \psi \rangle \geq 0, \forall \psi \in \mathcal{H}.
\end{align}
\begin{align}
    E(\Omega)= \rho(\Omega)^{-1/2}\rho(\Omega)\rho(\Omega)^{-1/2}=\id
\end{align}
\end{proof}

\begin{proposition}
Given a POVM $E : \Sigma \rightarrow \mathcal{L}(\mathcal{H})$ and a state $\sigma$, the mapping $$\rho(\Lambda)=\sigma^{1/2}E(\Lambda)\sigma^{1/2}$$ is a DOVM.
\label{Prop.FromPOVMtoDOVM}
\end{proposition}
\begin{proof}
The proof of this proposition is analogous to the previous one - it is a direct consequence of the definitions. As a matter of fact,  
\begin{align}
    \langle \psi, \rho(\Lambda) \psi \rangle &= \langle \psi, \sigma^{1/2}E(\Lambda)\sigma^{1/2} \psi \rangle \nonumber \\
    &=    \langle  \sigma^{1/2} \psi, E(\Lambda)\sigma^{1/2} \psi \rangle \geq 0.
\end{align}
\begin{align}
    \mbox{Tr}[\rho(\Omega)]&= \mbox{Tr}[\sigma^{1/2}E(\Omega)\sigma^{1/2}] = \mbox{Tr}[\sigma^{1/2} \id \sigma^{1/2}] \nonumber \\
    &= \mbox{Tr}[\sigma]=1.
\end{align}
\begin{align}
    \rho\left(\bigsqcup_{j}\Lambda_{j}\right)&= \sigma^{1/2}E\left(\bigsqcup_{j}\Lambda_{j}\right)\sigma^{1/2} \nonumber \\
    &=\sigma^{1/2}\sum_{j}E\left(\Lambda_{j}\right)\sigma^{1/2} \nonumber \\
    &=\sum_{j}\sigma^{1/2}E(\Lambda_{j})\sigma^{1/2}=\sum_{j}\rho(\Lambda_j)
\end{align}
\end{proof}

The following theorem represents our hybrid quantum extension of the agreement theorem. Its proof is entirely based on Aumann's original argument, especially on lemma \ref{Lemma.SetTheoryLemma}. We will provide further comments on this aspect later on.

\begin{thm}[Quantum Agreement Theorem]
Let $(\Omega,Q_1,Q_2,...,Q_N,\Sigma)$ be a knowledge model and $\rho:\Sigma \rightarrow \mathcal{L}(\mathcal{H})$ be a DOVM over $\Omega$. Define 
\begin{equation}
    E:=\bigcap_{i \in [N]}\{\omega \in \Omega| \,\, \rho_{|Q_{i}(\omega)}=\sigma_{i}\},
\end{equation}
with $\sigma_{i}$ a density operator acting on $\mathcal{H}$. If Tr$[\rho(C(E))] \neq 0$, then 
\begin{equation}
\sigma_1=\sigma_2=...=\sigma_N=\rho_{|C(E)}.
\end{equation}
\label{Thm.QuantumAumannTheorem}
\end{thm}
\begin{proof}
\begin{align}
    \rho_{|C(E)}&=\frac{\rho(C(E))}{\mbox{Tr}[\rho(C(E))]}= \frac{\rho\left(  \bigsqcup_{j \in [k_{i}]}D_{i}^{j}  \right)}{\mbox{Tr}\left[\rho\left(  \bigsqcup_{j \in [k_{i}]}D_{i}^{j}  \right) \right]}   \\
    &= \frac{\sum_{j \in [k_i]}\rho(D_{i}^{j})}{\mbox{Tr}\left[ \sum_{j \in [K_i]}\rho(D_{i}^{j}) \right]} =\frac{\sum_{j \in [k_i]}\rho_{|D_{i}^{j}}\mbox{Tr}\left[ \rho(D_{i}^{j}) \right]}{\sum_{j \in [k_i]}\mbox{Tr}\left[ \rho(D_{i}^{j}) \right]}  \\
    &=\sigma_{i}\frac{\sum_{j \in [k_i]}\mbox{Tr}\left[ \rho(D_{i}^{j}) \right]}{\sum_{j \in [k_i]}\mbox{Tr}\left[ \rho(D_{i}^{j}) \right]}  \\
    &= \sigma_{i}.
\end{align}
As the choice of index $i \in [N]$ has been entirely arbitrary, we can conclude that 
\begin{equation}
    \sigma_1=\sigma_2=...=\sigma_N=\rho_{|C(E)}.
\end{equation}
\end{proof}

Recall that this is a hybrid quantum-classical theorem. Instead of assigning classical probability distributions to events (or propositions), agents are allowed to use a richer mathematical object---a DOVM. Still, we are using the very same classical knowledge model---based on sigma-algebras and partitions of a given set---of def. \ref{Def.KnowledgeModel} as a proxy for each agent's enquiring/learning model. This is a significant reason why Aumann's argument extends to the quantum case. 

In more detail, two important properties carry over from the classical case to this hybrid quantum-classical case.  First, we have a notion of (quantum) probability defined over the measure space that is part of the underlying knowledge model, and this notion of probability behaves like an affine function on exclusive events.  This is the role of the DOVM.  Second, we have a definition of conditional states that satisfies an analogue of the law of total probability, i.e., for any partition $Q$, $\rho(\Omega) = \sum_{\Lambda \in Q}\rho_{|\Lambda} \Tr [\rho(\Lambda)]$. 

In summary, even though we have provided agents with a supposedly more powerful object to reason about events (or propositions), they are still limited by the inability to agree to disagree when common knowledge is present. The following section demonstrates that, contrary to recent arguments in the literature \cite{CTGAAP21}, agents with access to resources beyond quantum mechanics are also subject to Aumann's impossibility theorem. In this sense, the agreement theorem cannot be viewed as a physical principle that separates certain classes of generalised probability theories.

\section{Aumann's Theorem in GPTs}\label{Sec.GPTArgument}  

The next step in our argument involves proving that Aumann's agreement theorem is also valid in generalised probabilistic theories (GPTs) \cite{JH14,Plavala23}, a more general framework that encompasses classical probability and quantum theory as special cases. The upshot is that whenever such a theory is well-behaved concerning exclusive events and comes equipped with a well-defined notion of conditioning, then the agreement theorem must hold. 

We start by briefly reviewing exactly what we need from a GPT, and we suggest \cite{Plavala23} for a more comprehensive introduction to the subject. We then proceed with two additional definitions: state-valued measure (SVM) and conditional states. We conclude this section with our generalisation of the agreement theorem.  

\begin{definition}[Generalised Probability Theory]
A finite-dimensional \emph{generalised probability theory} (GPT) consists of a  triple $\mathcal{G}=(V,V^{+},u)$ where $V$ is a finite dimensional vector space, $V^{+}$ is a proper cone (convex, closed, pointed and generating) and $u$ is an order unit in the dual cone $(V^{+})^{\ast}:=\{f \in V^{\ast}; \,\, f(x) \geq 0; \forall x \in V\}$. Given $\mathcal{G}=(V,V^{+},u)$, its associated state space $\mathcal{S} \subseteq V^{+}$ is defined via 
    \begin{equation}
        \mathcal{S}:=\{\mu \in V^{+}| u(\mu)=1\}; 
    \end{equation}
and its set of effects $\mathcal{E}(\mathcal{S}) \subseteq (V^{+})^{\ast}$ is defined via: 
     \begin{equation}
         \mathcal{E}(\mathcal{S}):=\{\varphi: V^{+} \rightarrow \mathbb{R} | \,\, \varphi \,\, \mbox{is linear and} \,\, 0 \prec \varphi \prec u \},
     \end{equation}
 where $\varphi \prec \psi$ means $\varphi(x) \leq \psi(x)$ for every $x$ in $V^{+}$.    
\label{Def.GPT}
\end{definition}

In other words, generalised probabilistic theories form a framework used to capture the essence of states and measurements in a leaner and cleaner sense. Classical probability theory, as well as quantum theory and Box-world, are typical examples of GPTs. 

\begin{definition}[SVM]
Let $\mathcal{M}=(\Omega, \Sigma)$ be a measurable space and $(V,V^{+},u)$ be an arbitrary GPT. A \emph{state valued measure} (SVM) over $\mathcal{M}$ is a measurable function $\mu: \Sigma \rightarrow V^{+} $ satisfying:
\begin{itemize}
    \item [(i)] $\mu(\Omega) \in \mathcal{S}$
    \item [(ii)] $\mu(\Lambda) \in V^{+}$, for all $\Lambda \in \Sigma$;
    \item [(iii)] $\mu\left( \cup_{j \in J} \Lambda_{j}\right)=\sum_{j \in J}\mu(\Lambda_{j})$  for any countable family $\{\Lambda_{j}\}_{j \in J}$ of disjoint subsets in $\Sigma$. 
\end{itemize}
\label{Def.SVMinGPT}
\end{definition}

\begin{definition}
Let $\mu: \Sigma \rightarrow V^{+}$ be an SVM over a measurable space $\mathcal{M}=(\Omega, \Sigma)$. For each $\Lambda \in \Sigma$, we define the object
\begin{equation}
    \mu_{|\Lambda}:=\frac{\mu(\Lambda)}{u\left[ \mu(\Lambda) \right]}
    \label{Eq.DefConditionalFromSVM}
\end{equation}
as the conditional state of the SVM $\mu: \Sigma \rightarrow V^{+}$ with respect to the event $\Lambda$.
\label{Def.ConditionalFromSVM}
\end{definition}

At this stage, it should be clear where we are heading. Recall that two main things were central to the machinery involved in proving the agreement theorem. First, we need a notion of (generalised) probability defined over the measure space that is part of the underlying knowledge model, which behaves like an affine function on exclusive events. This is provided by the notion of an SVM.  Second, we need a notion of conditioning that satisfies a generalisation of the law of total probability, which in this case is, for any partition $Q$, $\mu(\Omega) = \sum_{\Lambda \in Q} \mu_{|\Lambda} u [ \mu(\Lambda)]$. Naturally, Aumann's original framework inherently possesses these two aspects. In the hybrid quantum case, defs. \ref{Def.DOVMs} and \ref{Def.ConditionalFromDOVM} granted quantum theory with those central features. Similarly, as we will see below, the defs. \ref{Def.SVMinGPT} and \ref{Def.ConditionalFromSVM} play the same fundamental role in our version of the agreement theorem for GPTs.

\begin{thm}[Agreement Theorem in GPTs]
Let $(\Omega,Q_1,Q_2,...,Q_N,\Sigma)$ be a knowledge model and $\mu:\Sigma \rightarrow V^{+}$ be an SVM over $\Omega$. Define 
\begin{equation}
    E:=\bigcap_{i \in [N]}\{\omega \in \Omega| \,\, \mu_{|Q_{i}(\omega)}=\mu_{i}\},
\end{equation}
where each $\mu_{i}$ is a state in $\mathcal{S}$. If $u[\mu(C(E))] \neq 0$, then 
\begin{equation}
\mu_1=\mu_2=...=\mu_N=\mu_{|C(E)}.
\end{equation}
\label{Thm.GPTAumannTheorem}
\end{thm}
\begin{proof}
\begin{align}
    \mu_{|C(E)}&=\frac{\mu(C(E))}{u[\mu(C(E))]}= \frac{\mu\left(  \bigsqcup_{j \in [k_{i}]}D_{i}^{j}  \right)}{u\left[\mu\left(  \bigsqcup_{j \in [k_{i}]}D_{i}^{j}  \right) \right]}   \\
    &= \frac{\sum_{j \in [k_i]}\mu(D_{i}^{j})}{u\left[ \sum_{j \in [k_i]}\mu(D_{i}^{j}) \right]} \\ 
    &=\frac{\sum_{j \in [k_i]}\mu_{|D_{i}^{j}}u(\left[ \mu(D_{i}^{j}) \right])}{\sum_{j \in [k_i]}u\left[ \mu(D_{i}^{j}) \right]}  \\
    &=\mu_{i}\frac{\sum_{j \in [k_i]}u(\left[ \mu(D_{i}^{j}) \right])}{\sum_{j \in [k_i]}u\left[ \mu(D_{i}^{j}) \right]}  \\
    &= \mu_{i}.
\end{align}
As the choice of index $i \in [N]$ has been entirely arbitrary, we can conclude that 
\begin{equation}
    \mu_1=\mu_2=...=\mu_N=\mu_{|C(E)}.
\end{equation}
\end{proof}

The fact that theorem \ref{Thm.GPTAumannTheorem}'s proof works exactly like the previous ones hints at something more profound. The agreement theorem should be considered more like a mathematical statement---telling us about probabilistic assignments and their conditioning defined over a specific construction in set theory (knowledge models)---than a signature of any physical restriction in the space of correlations. Related to that, because the structure of the proofs is practically the same, it is true that theorem \ref{Thm.QuantumAumannTheorem} follows from theorem \ref{Thm.GPTAumannTheorem} for the appropriate GPT. Also, in the case where all DOVMs commute with each other, the classical Aumann's agreement theorem—in its alternative version, as in theorem \ref{Thm.ClassicalAumannTheoremv2}—is a particular case of theorem \ref{Thm.QuantumAumannTheorem}. 

To conclude, we can say that in our hybrid GPT version of Aumann's agreement theorem, it is also true that agents cannot agree to disagree—provided they have common knowledge. A result showing that (i) although highly counter-intuitive, the impossibility of agreeing to disagree in the presence of common knowledge is also valid in any reasonable probability theory, regardless of how abstract or general they might be; and (ii) that the agreement theorem cannot be used as a criterion to separate physical theories. Nonetheless, it is essential to recognise that these conclusions come with the limitations inherent to the theory underlying them.


\section{Limitations}\label{Sec.Limitations}

To begin with, we emphasise that our argument is based on a standard set-theoretical model for knowledge, the same model adopted by Aumann in \cite{Aumann76}. Although that model does capture many elements of what we may want to mean by ``\emph{agent $a$ knows $E$}'', it falls short in capturing other essential aspects of (multi-) agent-centred knowledge. For example, in our framework, there is no room for communication between the agents. This fact is so significant that, within our model---and in Aumann's, for that matter---it might well be the case that agents reach common knowledge without exchanging any information---not even a bit. 
Knowledge models falling under the umbrella of dynamic epistemic logic are better adapted to deal with this type of problem. There, the concept of public announcements takes on a central role. When a true fact is publicly communicated, it imposes certain constraints on each agent's underlying knowledge. The fact is that after a finite number of public announcements, the agents end up in a state of common knowledge, and, as we saw, they are forced to have the same probabilistic description. See ref. \cite{DHK08} and refs. \cite{Demey14,Demey14b} for a critical introduction to the topic. The seminal work in ref. \cite{GP82} also addressed the agreement theorem from a dynamical standpoint, where the communication between the parts plays a significant role. 


As emphasized earlier, our result depends on assuming that the notion of conditioning obeys an appropriate generalization of the law of total probability.  For a GPT, this takes the form: for any partition $Q$, $\mu(\Omega) = \sum_{\Lambda \in Q} \mu_{|\Lambda} u [ \mu(\Lambda)]$, and in the special case of quantum theory it takes the form: $\rho(\Omega) = \sum_{\Lambda \in Q}\rho_{|\Lambda} \Tr [\rho(\Lambda)]$.  The classical notion of a conditional probability measure has this property and, from a mathematical point of view, we are free to make this a definitional property of what we mean by conditioning in a GPT.  But conditioning is not just a formal mathematical concept.  It is used in a variety of practical applications.  Is it feasible that this generalized law of total probability should apply in practice?

One of the main applications of conditional probability is the use of Bayesian conditioning to update probabilities in the light of new evidence.  If you initially assign a probability function $\mathbb{P}(H)$ and then learn the event $\Lambda$, Bayesian conditioning recommends that you should update your probability function to $\mathbb{P}_{\Lambda}(H) = \mathbb{P}(H|\Lambda)$.  If we replace the update rule with something else, then Aumann's agreement theorem would still hold provided $\sum_{\Lambda \in Q} \mathbb{P}(\Lambda) \mathbb{P}_{\Lambda}(H) = \mathbb{P}(H)$ for all partitions $Q$.  In the context of updating probabilities, this is called the \emph{reflection principle}.  Bayesian conditioning is not the only possible update rule that satisfies the reflection principle.  For example, let $0<q<1$ and set $\mathbb{P}_\Lambda(H) = q\mathbb{P}(H) + (1-q)\mathbb{P}(H|\Lambda)$.  This satisfies the reflection principle and might be an appropriate rule to adopt if you think that you are prone to hallucinating evidence for $\Lambda$ with probability $q$.  However, if you demand in addition that, for any partition $Q$ and any $\Lambda \in Q$, $\mathbb{P}_{\Lambda}(\Lambda) = 1$, then it is easy to show that the only possibility is Bayesian conditioning.  We call this principle \emph{evidential certainty}, as it means that your updated probabilities reflect certainty about the evidence you have collected.

So, classically, we could view both the reflection principle and evidential certainty as definitional properties of conditioning, and look for appropriate generalizations of them in quantum theory and GPTs.  However, generically, it is not possible to find a rule that has both properties.  Consider quantum theory and suppose that we measure a POVM $E:\Sigma \rightarrow \mathcal{L}(\mathcal{H})$ on a system prepared in the state $\rho$ and that we learn a partition $Q$.  The quantum version of the reflection principle would have $\rho$ update to states $\rho_{|\Lambda}$ that satisfy $\sum_{\Lambda \in Q}  \rho_{|\Lambda} \mathrm{Tr} [ E(\Lambda) \rho]= \rho$.  One way of satisfying this is to make use of the correspondence between POVMs and DOVMs (propositions \ref{Prop.FromDOVMtoPOVM} and \ref{Prop.FromPOVMtoDOVM}) and to define the update rule
\begin{align}
    \rho \rightarrow \rho_{|\Lambda} = \frac{1}{\mathrm{Tr} [E(\Lambda) \rho ]} \rho^{\frac{1}{2}} E(\Lambda) \rho^{\frac{1}{2}}.
\end{align}

The quantum version of evidential certainty would be $\mathrm{Tr} [E(\Lambda) \rho_{|\Lambda}] = 1$.  This is not possible to satisfy for a general POVM, but for a projective measurement, where each $E(\Lambda)$ is a projection operators, the L\"{u}ders rule

\begin{align}
    \rho \rightarrow \rho_{\Lambda} = \frac{1}{\mathrm{Tr} [E(\Lambda) \rho]} E(\Lambda) \rho E(\Lambda),
\end{align}
does satisfy evidential certainty.

However, it is not possible to find a rule that satisfies both the reflection principle and evidential certainty.  To see this, consider a qubit prepared in the state $\rho = \ket{+}\bra{+}$, where $\ket{+} = \frac{1}{\sqrt{2}} \left ( \ket{0} + \ket{1} \right )$ and suppose that we measure the projective measurement $E({0}) = \ket{0}\bra{0}, E({1}) = \ket{1}\bra{1}$, which is defined on the space $\Omega = \{0,1\}$.  Reflection requires that
\begin{align}
    \frac{1}{2} \rho_{0} + \frac{1}{2} \rho_{1} = \ket{+}\bra{+},
\end{align}
but this requires $\rho_{0} = \rho_{1} = \ket{+}\bra{+}$ because a pure state is an extremal point of the convex set of density operators.  On the other hand, evidential certainty requires $\rho_{0} = \ket{0}\bra{0}$ and $\rho_{1} = \ket{1}\bra{1}$ because these are the only states that give certainty for the projective measurement $E$.  Therefore, if we want to define a notion of conditioning in quantum theory, then the reflection principle and evidential certainty cannot \emph{both} be definitional properties of conditioning.  We have to make a decision about which one is most important. 

Given that the L\"{u}ders rule is usually thought of as the minimally disturbing state update rule for a quantum measurement, and that it satisfies evidential certainty but not reflection, you might be inclined to reject the reflection principle and try to generalize Aumann's theorem by replacing Bayesian conditioning with L\"{u}ders updating.  This was explored in \cite{Khrennikov15}, but adopting this strategy will clearly not lead to an agreement theorem.  This is because measurements necessarily disturb the state of the system in quantum mechanics, so we have to think of an agent measuring the system as an active intervention rather than a passive observation.  A subsequent measurement can completely invalidate the basis for knowledge gained by previous measurements.  

This is not particularly special to quantum mechanics, but would be true in a classical theory of invasive measurements as well.  For example, suppose that the procedure by which an agent learns that the state of the world $\omega$ is in the element $Q(\omega)$ of their partition $Q$ causes the state to be randomly permuted among all the states in $Q(\omega)$.  The agent's knowledge of $Q(\omega)$ would still be valid but it would destroy the basis for the knowledge that other agents previously gained by measuring the system.  Of course, there is no agreement theorem in this scenario, but this is because we are treating a complex intervention in a physical system as a pure updating of knowledge.  It is not really telling us about the structure of knowledge in the theory.

On the other hand, it has been argued \cite{fuchs2002} that the update rule $\rho \rightarrow \rho_{|\Lambda} = \frac{1}{\mathrm{Tr}[E(\Lambda) \rho]} \rho^{\frac{1}{2}} E(\Lambda) \rho^{\frac{1}{2}}$, which we call \emph{Fuchs' rule}, is a better quantum analogue of Bayesian conditioning.  While Fuchs' rule does not apply to the conventional quantum state of a system that is directly measured, it does apply in other causal scenarios that arguably have a better claim to be pure acquisition of knowledge.  For example, suppose that two systems $A$ and $B$ are in a joint state $\rho_{AB}$ on a tensor product Hilbert space $\mathcal{H}_A \otimes \mathcal{H}_B$.  If we measure a POVM $E:\Sigma \rightarrow \mathcal{L}(\mathcal{H}_A)$ on system $A$ then this can be thought of as an indirect acquisition of information about $B$, and it can be shown that the appropriate way to update the reduced state of system $B$ upon learning te event $\Lambda$ is $\rho_B \rightarrow \frac{1}{\mathrm{Tr}[F(\Lambda) \rho]} \rho_B^{\frac{1}{2}} F(\Lambda) \rho_B^{\frac{1}{2}}$, where $F:\Sigma \rightarrow \mathcal{L}(\mathcal{H}_B)$ is a POVM on system $B$ that is determined by $E$ and the structure of the correlations in $\rho_{AB}$.  So, if we think of measuring $E$ on $A$ as an indirect measurement of $F$ on $B$, then the state of $B$ updates according to Fuchs' rule.  Since this measurement does not involve any interaction with $B$, which could be located arbitrarily far away from $A$, this type of measurement has a claim to be pure knowledge acquisition about $B$.

Another scenario in which this rule applies is the following.  Alice can prepare a quantum system according to a DOVM $\rho:\Sigma \rightarrow \mathcal{L}(\mathcal{H})$ by first generating a sample from the classical probability measure $\mathbb{P}(\Lambda) = \mathrm{Tr}[\rho(\Lambda)]$ (by flipping coins, rolling dice, drawing cards, etc.)  Then, if she obtains the classical state $\omega$, she prepares the quantum system in the state $\rho_{|\{\omega\}} = \rho(\{\omega\})/\mathrm{Tr}[\rho(\{\omega\})]$.  Suppose Bob knows the details of this procedure, but does not know anything about the classical state that Alice obtained.  Then, Bob would assign the state $\rho = \rho(\Omega)$ to the system.  If Bob subsequently learns that $\omega \in \Lambda$, not by measuring the quantum system but by acquiring information about Alice's classical state, e.g., by talking to Alice, then he should update his state from $\rho$ to $\rho_{|\Lambda}$.  This update satisfies reflection and is in fact just another example of Fuchs' rule via the correspondence between POVMs and DOVMs given in \cref{Prop.FromDOVMtoPOVM}.  Here, too, there is a case to be made that this is a pure acquisition of knowledge because Bob obtains his information without directly interacting with the quantum system.

Finally, suppose that Alice prepares the system as just described, but that Bob acquires information by directly measuring the POVM $E:\Sigma' \rightarrow \mathcal{L}(\mathcal{H})$ on the system instead of just talking to Alice \footnote{The measurable space of Bob's measurement is denoted $(\Omega',\Sigma')$ to distinguish it from Alice's space $(\Omega,\Sigma)$ because the two spaces do not need to be the same.}.  As previously discussed, Fuchs' rule does not apply if Bob wants to update his state in order to predict what will happen to the system after the measurement.  But, as shown in \cite{LS13}, it is the correct rule to use for \emph{retrodicting} the past of the system, e.g., for inferring information about the classical variable that Alice used to decide which quantum state to prepare.   By \cref{Prop.FromDOVMtoPOVM}, Alice's DOVM can be converted into a POVM $F:\Sigma\rightarrow \mathcal{L}(\mathcal{H})$, which in this case we call the \emph{retrodictive POVM}.  Before he makes his measurement, Bob can calculate the probability that Alice's classical state is in $\Lambda$ by the formula $\mathbb{P}(\Lambda) = \Tr[\rho(\Lambda)] = \Tr[F(\Lambda) \rho]$, where $\rho = \rho(\Omega)$ is Bob's \emph{retrodictive state}.  If Bob learns that the outcome of his mesurement is in the set $\Lambda'$ then he should update the retrodctive state to $\rho_{|\Lambda'} = \frac{1}{\mathrm{Tr}[E(\Lambda') \rho]} \rho^{\frac{1}{2}} E(\Lambda') \rho^{\frac{1}{2}}$ and then update his probability for Alice's classical state to $\mathbb{P}_{\Lambda'}(\Lambda) = \mathrm{Tr}[F(\Lambda) \rho(\Lambda')]$.  The retrodictive state update is another instance of Fuchs' rule and can be viewed as a pure acquisition of knowledge because we are inferring information about the, presumably fixed, past.

The first two scenarios can easily be generalized to GPTs, but the retrodiction example does not hold in all GPTs, as it relies on the equivalence of the predictive and retrodictive formalisms for quantum theory, based on the equivalence between POVMs and DOVMS given in \cref{Prop.FromDOVMtoPOVM} and \cref{Prop.FromDOVMtoPOVM}, which does not generalize to all GPTs.

Since Aumann's theorem is supposed to be about the structure of common knowledge, rather than a series of chaotic disturbances of a system, update rules satisfying reflection are arguably a better match for generalizing agreement theorems than those that prioritize other features, such as evidential certainty.  However, this does not mean that our notion of conditioning is the uniquely correct generalization of conditioning to quantum theory.  Rather, the properties that define classical conditioning cannot all hold simultaneously in quantum theory, so there will be more than one generalization, and which one is appropriate to use depends on the application.  Therefore, showing that some notion of conditioning violates a theorem of classical probability theory does not imply that there is no such theorem in quantum theory.  It might hold for a different notion of conditioning, so you first have to argue that you are using the notion of conditioning that is most appropriate for the problem at hand.  In our opinion, using a notion of conditioning that preserves the relevant results of classical probability theory as far as possible is usually the best thing to do, as it helps us to focus on the necessary differences between classical and quantum as opposed to quirks of one specific notion of conditioning.




\section{Comparison to other works}\label{Sec.Comparison}

Because of its impact on how we should think about collective reasoning, Aumann's theorem has appeared several times in the specialised literature of quantum foundations. Paradoxically enough, the myriad of results indicates there is no consensus---or common knowledge. 

In refs. \cite{Khrennikov15,KB14}, the authors argue that, in general, quantum agents may evade the agreement theorem, which would show, therefore, a clear cut between quantum and classical strategies of reasoning---although they do investigate necessary conditions in which Aumann's result holds in a quantum-like framework. Their setup differs from ours in two main aspects. First, their quantum-like version of Aumann's theorem is based on operator lattices, very much in the spirit of Pitowski's original works \cite{Pitowsky89}. The very notions of 'states of the world', 'knowledge' and 'common knowledge' are, therefore, defined in terms of pure states, projectors and the relationship between them. Granted, their generalisation are stated in terms of standard probability distributions, and they use the Born rule to translate back from operators to real numbers, but the core of their epistemic machinery differs from ours. Second, and because of their focus on standard probability, their notion of conditioning is defined via the usual update rule, which differs from ours---see def. \ref{Def.ConditionalFromDOVM}. Whereas we maintain the underlying epistemic structure intact and generalise the notion of probabilistic assignments, they generalise the underlying structure while retaining the classical assignments. Given the dependency on operator lattices, it is not clear to us how one could generalise the results of \cite{Khrennikov15,KB14} to GPTs.

Before concluding, there is another generalisation of the agreement theorem that we want to compare our results with: the works of Contreras-Tejada et.al. \cite{CTGAAP21, BrandenburgerEtAl24}. In those works, the authors prove that the impossibility of agreeing to disagree may constrain classical and quantum theories, potentially separating them from more general probability theories. In \cite{CTGAAP21}, they argue that an extremally non-signalling box, the PR-box, allow for agents to agree to disagree on common certainty. To do so, they first reframe Aumann's original formulation into the usual black-box correlation scenario \cite{NCPSW14}. But this reframing comes with a cost. Although one can easily construct a non-signalling correlation from a knowledge model (with a probability function over it), for the converse of this construction to work for all non-signalling boxes, one should allow quasi-probability measures over the knowledge model. In other words, if we start from correlation scenarios, to translate back to Aumann's original formulation, we should consider negative probabilities---particularly in \cite{BrandenburgerEtAl24}, where the authors use signalled probability measures as their starting point. Besides this reframing, they propose a notion of common certainty between observers. Their definition mimics the original notion of common knowledge but differs from it. In fact, the certainty modality that the authors explore in detail there (ref. \cite{CTGAAP21}) is weaker than the knowledge modality we explore here and, consequently, common certainty is weaker than common knowledge.

In a bipartite scenario, where Alice has access to $\mathcal{X}$ inputs and $\mathcal{A}$ outputs and Bob has $\mathcal{Y}$ inputs and $\mathcal{B}$ outputs, according to Contrera-Tejadas et. al. \cite{CTGAAP21}, common certainty at, say, $(a=0,b=0,x=0,y=0)$ of Alice assigning $q_A$ to $F_{B}=\{(1,b,1,y); b \in \mathcal{B} \mbox{ and } y \in \mathcal{Y} \}$ and Bob assigning $q_B$ to $F_{A}=\{(a,1,x,1); a \in \mathcal{A} \mbox{ and } x \in \mathcal{X} \}$ holds true whenever $(a=0,b=0,x=0,y=0) \in A_{n} \cap B_{n}, \forall n \in \mathbb{N}$, where:
\begin{align}
    A_{n}:= \alpha_{n} \times \mathcal{B} \times \mathcal{X} \times \mathcal{Y}   \\
    B_{n}:=  \mathcal{A} \times \beta_{n} \times \mathcal{X} \times \mathcal{Y} 
\end{align}
and 
\begin{align}
   & \alpha_{0}=:\{a \in \mathcal{A}; q_{a}=p(b | a, x=0, y=1)\} \nonumber  \\
   & \beta_{0}=:\{b \in \mathcal{B}; q_{b}=p(a | b, x=1, y=0)\} \\
   & \alpha_{n+1}:=\{ a \in \alpha_{n}; p(B_n | a, x=0, y=0) =1 \}, \forall n \in \mathbb{N} \nonumber  \\
   & \beta_{n+1}:=\{ b \in \beta_{n}; p(A_n | B, x=0, y=0) =1 \}, \forall n \in \mathbb{N}  \nonumber, 
\end{align}
with $p(a,b|x=1,y=1)=0, \mbox{ for every} \,\, a \neq b$, meaning that $F_A$ and $F_B$ are perfectly correlated. Although their definition brings forth a hierarchy reminiscent of Aumann's original common knowledge hierarchy (def. \ref{Def.CommonKnowledge}), Contrera-Tejada et. al.'s standpoint depends not only on the perfect correlation of $F_A$ and $F_B$, but also on the inputs and outputs chosen by the observers (with their guesses of each other's inputs) and on the fact that they attribute probability one to certain events in the hierarchy---all of which are not required in our framework. 

Finally, in \cite{CTGAAP21}, the authors also assume that agents describe the outcomes of their experiments with standard probability theory, and consequently, are also bound to use the Born rule in the quantum case, with the conditioning rule left unspecified. Recall that even though we kept Aumann's original knowledge model fixed, with DOVMs and SVMs we have been able to go beyond the paradigm of standard probability theory, and were also able to define a conditional state---see defs. \ref{Def.ConditionalFromDOVM} and \ref{Def.ConditionalFromSVM}.

Before moving to the concluding remarks, we would like to briefly mention two recent works that also explore Aumann's agreement scenario. In ref. \cite{DBS26}, with a setup that investigates whether two agents represented by separate laboratories can or cannot maintain differing probability estimates of a shared quantum property of interest, the authors propose that common certainty of disagreement (CCD) is genuinely quantum. By doing so, they conclude that their framework may offer a rigorous approach to the longstanding issue of understanding intersubjectivity across agents in quantum theory. Because intersubjectivity across agents in quantum theory may be seen as a cornerstone in some Wigner's Friend's scenarios \cite{SRD26}, we also want to mention ref. \cite{CDiB26}. There, the authors derive an operational version of Aumann's agreement theorem without assuming an objective state of the world. Agreeing with our results, this move allows them to assert the theorem's validity in quantum theory and, especially, in situations with indefinite causal order. The authors observe, however, that in situations where no joint probability distribution over the relevant outcomes can be meaningfully assigned, the theorem therefore ceases to apply; which may give reason why, in Wigner's Friend's Scenarios, rational agents may agree to disagree. 

\section{Conclusions}\label{Sec.Conclusion}
 
In this work, we show that Aumann's impossibility of agreeing to disagree is characteristic of any generalised probability theory, as long as we maintain the standard knowledge model formulation unchanged

The reason Aumann's original argument carries over to more general theories stems from: $(i)$ the use of a countably additive probability-like measure defined over the same set of states of the world $\Omega$ and the same sigma-algebra $\Sigma$ which, in turn, defines the underlying knowledge model $(\Omega, Q_1,...,Q_N,\Sigma)$; and $(ii)$ an appropriate definition of conditioning. In this sense, we believe the agreement theorem should be viewed more as a statement about probability theories defined over knowledge models, rather than a criterion for delineating different physical theories.  

We have also commented on the lack of room for communication between the agents in our formulation---and in Aumann's, for that matter. It is clear that communication should play a crucial role in any notion of agreement, and that has been emphasised in contemporary versions of the original theorem \cite{BrandenburgerEtAl24,Demey14b}. Nonetheless, if we are to focus on the dynamics of epistemic notions in multi-agent scenarios with multiple rounds of private or public announcements, we enter the realm of dynamic epistemic logic. Although a classical reformulation of Aumann's result has already been investigated within this logical framework, a generalisation, like the one we did here, could confirm what we mentioned before, that the agreement result is just a statement about probability theories defined over knowledge models, not a separation criterion for physical theories. 

Even though we have extended the agreement theorem in the direction of generalised probabilistic descriptions, there is yet another possible path to generalise Aumann's theorem. One could have started from the concept of (incompatible) physical operations and, consequently, in the sense of Randall and Foulis, substituted both the knowledge model and the probabilistic assignments to something more akin to a test space \cite{Wilce24}. It may be the case that the overlap between the physical operations can allow for the possibility of agreeing to disagree. This alternative path to generalise the agreement theorem should be thoroughly investigated in the future.

Finally, we would like to highlight a potential connection between common knowledge and the emergence of an objective reality. Although this is highly speculative, we have mentioned \textit{en passant} this possibility in the introduction, and we want to conclude our contribution by emphasising this point again. The lesson learned from Wigner's paradox, mainly as discussed in refs. \cite{Cavalcanti2021, BongEtAl20}, is that, in general, there is a discrepancy between Wigner's external perspective and his Friend's internal perspective. The asymmetry between the two agents results in Wigner's impossibility of assigning a definite outcome for events inside the Friend's laboratory, which would lead us to abandon either Local Action or Absoluteness of Events. In ref. \cite{Cavalcanti2021}, the author hints about a `law of thought' that might reconcile the two perspectives. We suggest that common knowledge could be that law of thought the author was looking for. Common knowledge would be the mechanism responsible for the emergence of a collective notion of objective reality, a reality where each and every agent would agree on the descriptions of events. A first step in addressing this point was taken by the authors of ref. \cite{PoderiniEtAl23}, who discussed a notion of collective objectivity (commonly agreed-upon reality) in connection with Quantum Darwinism.

\begin{acknowledgments}
Cristhiano Duarte is undoubtedly grateful to Eric Cavalcanti, whose discussions emphasised that Aumann's theorem lacked a proper communication dimension. CD also thanks the hospitality of the Institute for Quantum Studies at Chapman University. ML was supported, in part, by the Fetzer Franklin Fund of the John E. Fetzer Memorial Trust.  This research was also supported by the Fetzer Franklin Fund of the John E.\ Fetzer Memorial Trust and by grant number FQXi-RFP-IPW-1905 from the Foundational Questions Institute and Fetzer Franklin Fund, a donor-advised fund of Silicon Valley Community Foundation. This work was supported by CNPq through a grant from the Conhecimento Brasil Program (Linha 1 and Linha 2).  \textbf{Authors' Contribution:} Both authors contributed equally to the whole research project. 
 
\end{acknowledgments}

\printbibliography


\end{document}